\documentclass[12pt]{article}
\usepackage[pdftex]{color}
\usepackage{amsmath}
\usepackage{amssymb}
\usepackage{graphicx}
\usepackage{mathrsfs}
\usepackage{subfigure}
\usepackage[longnamesfirst]{natbib}
\usepackage{verbatim}
\usepackage{tikz}
\usetikzlibrary{trees}

\newtheorem{theorem}{Theorem}[section]

\newtheorem{lemma}[theorem]{Lemma}
\newtheorem{definition}[theorem]{Definition}
\newtheorem{remark}[theorem]{Remark}
\newtheorem{example}[theorem]{Example}

\newcommand{\rref}[1]{{\rm (\ref{#1})}}

\newcommand{\esssup}{\operatorname*{\mathrm{ess\,sup}}}
\newcommand{\essinf}{\operatorname*{\mathrm{ess\,inf}}}

\newcommand{\cA}{{\cal A}}

\newcommand{\cF}{{\cal F}}
\newcommand{\cG}{{\cal G}}

\newcommand{\qed}{\mbox{ }~\hfill~$\Box$ \vspace{1ex} }

\newenvironment{proof}{\noindent{\sc Proof: }}{ \qed }

\newenvironment{proof3}[1]{\noindent{\sc Proof of Lemma #1: }}{ \qed  }

\newcommand{\rmII}{\text{\it I\kern-.08em I\,}}
\newcommand{\rmIII}{\text{\it I\kern-.08em I\kern-.08em I\,}}
\newcommand{\rmIV}{\text{\it I\kern-.08em V\,}}

\begin{document}

\title{The Foster--Hart  Measure of Riskiness for General Gambles}
\author{
Frank Riedel\thanks{This paper has been written while I was visiting Princeton University. Special thanks to Patrick Cheridito for hospitality.  Financial Support through the German Research
Foundation, International Graduate College ``Stochastics and Real World Models'', Research Training Group EBIM, ``Economic Behavior and Interaction Models''
  is gratefully acknowledged.} \and Tobias Hellmann \thanks{Institute of Mathematical Economics (IMW), Bielefeld University and Bielefeld Graduate School of Economics and Management (BiGSEM). Financial support provided by BiGSEM is gratefully acknowledged.}
}

\date{\today}
\maketitle

\begin{abstract}
\noindent Foster and Hart proposed an operational measure of riskiness for discrete random variables. We show that their defining equation has no solution for many common continuous distributions including many uniform distributions, e.g. We show how to extend consistently the definition of riskiness to continuous random variables. For many continuous random variables, the risk measure is equal to the worst--case risk measure, i.e. the maximal possible loss incurred by that gamble.  We also extend the Foster--Hart risk measure to dynamic environments for general distributions and probability spaces, and we show that the extended measure avoids bankruptcy in infinitely repeated gambles.
\end{abstract}

{\footnotesize{ \it Key words and phrases:} Risk Measures, Operational, Bankruptcy, Continuous Random Variable \\
{\it \hspace*{0.6cm} JEL subject classification: } D81, G11}

\newpage

\section{Introduction}

\cite{FoHa09} introduce a notion of riskiness, or critical wealth level, for gambles with known distribution. The measure is objective in so far as it depends only on the distribution of the outcome; in decision--theoretic terms, it is probabilistically sophisticated, in the language of the finance literature on risk measures, it is law--invariant. The concept  admits a simple operational interpretation because an agent avoids bankruptcy in the long run almost surely provided he accepts a gamble  only if his current wealth exceeds the critical value.

\cite{FoHa09} study only gambles with finitely many outcomes; as many financial applications involve distributions with simple densities like the uniform or lognormal distribution, it seems natural and important to generalize the concept of critical wealth level to such cases. To our surprise, we realized that even for the most simple case of a uniform distribution, the defining equation of Foster and Hart does not always have a finite solution. The non--finite value of infinity is always a possible solution for the defining equation, but it would seem most counterintuitive and implausible to reject a uniformly distributed gamble on, say, the interval $[-100,200]$ at arbitrary wealth levels.

We thus set out to study the concept of riskiness for distributions with densities. We show that there are two classes of gambles. For some of them, the defining equation of Foster and Hart has a finite solution, and one can use this number as its riskiness. For others, like the uniform one described above, the defining equation has no solution. We show that in this case, the right notion of riskiness is the \emph{maximal loss} of the gamble. This might seem surprising at first sight, as for finite gambles, the Foster--Hart riskiness is always strictly larger than the maximal loss. But we show that it is the only rational way to extend the concept to arbitrary gambles.

We approach the problem by approximating non--finite gambles by finite ones. We show that the corresponding sequence of critical wealth levels converges. For the uniform distribution above, the approximation involves a sequence of uniformly distributed finite gambles on a grid. As the grid becomes finer and finer, we show that the riskiness values converge indeed to the maximal loss. This result carries over to gambles with arbitrary distributions with a strictly positive density on their support. 

For gambles with a density where the Foster--Hart equation does have a solution, we also show that the riskiness values converge to that value. This is important as it provides a justification for using finite gambles with many outcomes as an approximation to non--finite gambles with a density.

We also provide another justification of  our extended definition of riskiness  by extending the no--bankruptcy result of Foster and Hart. Indeed, an investor who accepts gambles only if his wealth exceeds the critical wealth level will never go bankrupt. This implies, in particular, that an investor who accepts the above uniform gamble on $[-100,200]$ when his wealth exceeds $100 \$ $ will not go bankrupt. 

In order to prove the no--bankruptcy result in a formally correct way, we introduce the notion of the \textit{dynamic Foster--Hart riskiness} which is the dynamic extension of the static Foster--Hart riskiness.
Dynamic measurement of risk plays an important role in the recent literature \footnote{See, among others, \cite{DeSc05} and \cite{FoeSch11}, Chapter 11 for a detailed introduction to dynamic risk measures.} since it allows, in contrast to the static case, to measure risk of financial positions over  time. The arrival of new information can thus  be taken into account. This is important for many situations; for instance, if one faces a gamble that has its payments in, say, one month, there may arrive  more information in two weeks from now and the risk assessment can then be determined in a more precise way. Thus, this extension provides an important tool for the purpose of application of the Foster--Hart riskiness in a dynamic framework.

Our paper  provides also  a new and more general proof of the no--bankruptcy theorem. In particular, we remove the artificial assumption that all gambles are multiples of a finite number of basic gambles.

The paper is set up as follows. In the next section, we illustrate the motivating problem with the Foster--Hart equation. In Section 3, we approximate gambles with densities by discrete gambles. This allows us to use the original Foster--Hart definition of riskiness, and we study the limit behavior for the approximating sequences.  Section 4 contains the conditional version of the riskiness concept, and proves the no bankruptcy theorem.
Section 5 contains examples. 

\section{The Foster--Hart Riskiness and a Motivating Example}

For a random variable $X$ on some probability space $(\Omega,\cF,P)$ that satisfies $EX>0$ and $P(X<0)>0$, the  Foster--Hart riskiness is given by $R(X)=\frac{1}{\lambda},$ where $\lambda$ is the unique positive solution of the equation
\begin{equation}\label{EqnDefRiskiness}
  E \log\left(1+\lambda X\right) = 0\,.
\end{equation}

Note that for discrete random variables this equation is defined for all nonnegative values of $\lambda$ up to, but strictly smaller than  $\lambda^*(X)=1/L(X)$, where  $L(X)=\max_{\omega\in\Omega} (-X (\omega))$ is the maximal loss of the gamble.
For discrete random variables with positive expectation and possible losses, such a strictly positive solution always exists. For example, if $X$ is a Bernoulli random variable with
$$P(X=200)=P(X=-100)=\frac{1}{2}\,,$$
one can easily verify that
$0=\frac{1}{2} \log\left(1+200\lambda\right) + \frac{1}{2} \log\left(1-100\lambda\right)$
leads to
$\lambda=1/2000$ or $R(X)=2000$.

The starting point of our analysis is the following simple observation that struck us when we wanted to apply the measure of riskiness to more general, continuous distributions. Let $X$ be uniformly distributed on $[-100,200]$, say.  $X$ has a positive expectation of $50$ and incurs losses with positive probability, so it qualifies as a gamble in the sense of Foster and Hart.  Nevertheless, the equation \rref{EqnDefRiskiness} has no positive solution!

\begin{example}
\label{MotivatingEx}
Let $X$ be uniformly distributed over $[-100,200]$. $X$ has positive expectation and positive probability of losses.
  We plot the function $\phi(\lambda)= E \log\left(1+\lambda X\right)$ over the interval $[0,1/100]$ in Figure \ref{FigUniform}.  We see that no solution for $\lambda>0$ to \rref{EqnDefRiskiness} exists. To prove this, note that with $\lambda^*(X)=1/100$ we have   $$   E \log\left(1+\lambda^*(X) X\right)= \int_{-100}^{200} \frac{1}{300} \log\left(1+\frac{ x}{100}\right) dx = \log 3 - 1 \simeq 0.0986 > 0 \,.$$ As $\phi$ is concave and strictly increasing in $0$, we conclude that $\phi(\lambda)>0$ for all $\lambda \in (0,\lambda^*(X)]$.
\end{example}

\begin{figure}
 \centering
  \includegraphics[scale=0.5]{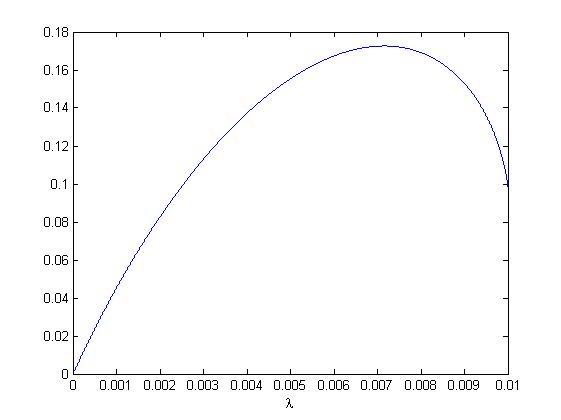}
  \caption{The function $\lambda \mapsto E \log\left(1+\lambda X\right)$ for the uniform distribution over $[-100,200]$ has no zero.}
  \label{FigUniform}
\end{figure}

How can we assign a riskiness to this simple gamble if the defining equation of Foster and Hart has no solution?
One could take $\lambda=0$, of course, resulting in a riskiness measure of $\infty$. Does this mean that one should never accept uniformly distributed gambles? Then an investor with a wealth of, say, a billion dollars would reject the above uniform gamble on $[-100,200]$.  Given that such a gamble has a nice expected gain of 50 and a maximal loss of 100, this would seem quite counterintuitive.

 In this note, we set out to extend the notion of riskiness for continuous random variables like the uniform above by approximating them via discrete random variables. We show that the limit of the riskiness coefficients exists. If the expectation
$E \log(1+\lambda^*(X) X ) $ is negative (including negative infinity), one can use the equation \rref{EqnDefRiskiness} to define the riskiness of $X$. (This also shows that our notion is the continuous extension of the discrete approach). For continuous random variables with  $E \log(1+\lambda^*(X) X ) \geq 0 $ such as our uniform random variable above, we use the limit of the riskiness coefficient of the approximating discrete random variables. This limit  turns out to be  equal to the maximal loss $L(X)$.

Whereas the riskiness measure is quite conservative for Bernoulli random variables as it prescribes a high value of $2000$ for the wealth for a Bernoulli random variable with maximal loss of $100$, it does accept the uniform random variable over $[-100,200]$, which has the same maximal loss of $100$, even if one has just $100$ \$.

How can one explain this stark difference? The point is that the Bernoulli random variable above is quite far away from the uniform random variable over the whole interval $[-100,200]$. For the Bernoulli case, a loss of $100$ \$ has a high probability of $50$ \%.  For the uniformly distributed random variable, a loss of close to $100$ \$ occurs with a very small probability and the loss of exactly $100$ \$ even with probability zero; as the defining aim of the operational measure of riskiness is to avoid bankruptcy, this is a crucial difference. Indeed, our analysis below shows that the riskiness decreases if we spread the Bernoulli random variable more uniformly over the interval $[-100,200]$, say by using a uniform grid. For discrete random variables uniformly distributed over a fine grid, the riskiness is close to  the maximal loss of $100$.

\section{Foster--Hart Riskiness for General Continuous Distributions}

In the remaining part of this note, we call a gamble a random variable $X$ on some probability space $(\Omega,\cF,P)$ with $EX>0$ and $P(X<0)>0$ that has either a finite support or admits a strictly positive density on its support which can either be the compact interval $[-L,M]$ or $[-L,\infty)$.

For a gamble $X$ let us study the function
$$\phi(\lambda)=  E \log\left(1+\lambda X\right) $$ which is well--defined for all $\lambda \in [0,L(X))$, where we remind the reader that we defined the maximal loss $L(X)=\max (-X)$ whenever $X$ is a discrete random variable and we define $L(X)=\esssup (-X)$ \footnote{We define as usual $\esssup (-X):=\inf\{x \in \mathbb R | P(-X>x)=0 \}$.} if $X$ is continuous.

The function $\phi(\lambda)$ is continuous and strictly concave, with strictly positive slope in $0$ as $E X>0$ (see the argument in \cite{FoHa09}). If $X$ is discrete, its distribution places a strictly positive weight on the event $\{X=-L(X)\}$, where $X$ achieves the maximal possible loss. As a consequence,  the expression $\log\left(1-\lambda L(X)\right)$ tends to minus infinity as $\lambda$ approaches the  value $\lambda^*(X)=1/L(X)$. The expectation that defines $\phi$ then also converges to negative infinity for $\lambda \to \lambda^*(X)$, and the function $\phi$ has a unique zero in $(0,\lambda^*(X))$.

For continuous random variables $X$, as the uniform one, the event $X=-L(X)$ has probability zero. The expectation $ E \log\left(1+\lambda X\right)$ is then well--defined even for $\lambda=\lambda^*(X)$, possibly with a value of negative infinity. When this last expression is positive, however, equation \rref{EqnDefRiskiness} has no positive solution. This is indeed the case for suitable uniform distributions as we showed above.

In the following we extend the riskiness in a rational way such that we can apply it to such continuous random variables. For this purpose we approximate them by discrete random variables and study the asymptotic behavior of their riskiness number.

Let us describe our strategy more formally. We approximate $X$ from below by an increasing sequence of discrete random variables $X_n \uparrow X$. For each $X_n$, there is a unique positive number $0<\lambda_n < \lambda^*(X)$ that solves the defining equation
$$E \log\left(1+\lambda_n X_n \right)=0\,.$$
The sequence $(\lambda_n)$ is increasing and bounded, thus the limit $\lambda_\infty=\lim_{n\to\infty} \lambda_n$ exists and is bounded by $\lambda^*(X)$. We then show that for random variables $X$ with
$\phi(\lambda^*(X))=E \log\left(1+\lambda^*(X)X\right)<0$, $\lambda_\infty$ is the unique positive solution of the defining equation \rref{EqnDefRiskiness}. For random variables like the uniform one, where
$\phi(\lambda^*(X))=E \log\left(1+\lambda^*(X)X\right)\ge 0$, we have $\lambda_\infty=\lambda^*(X)$.

Without loss of generality, we take $L=1$ and thus $\lambda^*(X)=1$. We consider two different sequences of partitions of the support of the continuous random variables. On the one hand, if the support of $X$ is the compact interval $[-1,M]$, we define
$$x^n_k = -1 + \frac{k}{2^n} (M+1), \qquad k=0,\ldots,2^n-1\,$$ and set
$$X_n = \sum_{k=0}^{2^n-1} x^n_k 1_{\left\{x^n_k \le X < x^n_{k+1}\right\}} \,.$$

On the other hand, if the support of $X$ is the infinite interval $[-1,\infty)$, we define
$$x^n_k = -1 + \frac{k}{2^n}, \qquad k=0,\ldots,n2^n-1\,$$ and set for $n \geq 1$
$$X_n = \sum_{k=0}^{n2^n-1} x^n_k 1_{\left\{x^n_k \le X < x^n_{k+1}\right\}} +(-1+n) 1_{\{X \geq (-1+n)\}}.$$

For both cases the next Lemma holds true.

\begin{lemma}\label{LemmaIncreasingApprox}
  The sequence $(X_n)$ is increasing and $\lim X_n = X$ $a.s.$
\end{lemma}

This is a well--known measure--theoretic construction. The proof is reported in the appendix.


As we have $-1\le X_n\le X \in L^1$,\footnote{We denote by $L^1$ the space of all random variables $X$ with $EX<\infty$.} we conclude, using Lebesgue dominated convergence theorem, that $\lim E X_n = E \lim X_n = E X >0$. Hence, for $n$ sufficiently large, we have $E X_n >0$. From now on, we always look at such large $n$ only.   As the density of $X$ is strictly positive on its support, we also have
$P(X_n<0)\ge P(X_n=-1)>0$. Therefore, the Foster--Hart riskiness is well--defined for $X_n$. Let $\lambda_n \in (0,1)$ be the unique positive solution of
$$E \log(1+\lambda_n X_n)=0\,.$$
The next Lemma follows directly by Lemma \ref{LemmaIncreasingApprox} and by the monotonicity of the Foster--Hart measure of riskiness, see Proposition 2 in \cite{FoHa09}.
\begin{lemma}
  The sequence $(\lambda_n)$ is increasing and bounded by $L(X)=1$. As a consequence, $$\lambda_\infty = \lim \lambda_n$$ exists
 and is less or equal to $L(X)=1$.
\end{lemma}

We can now state our main theorem. The limit $\lambda_\infty$ identified in the previous lemma is the right tool to define riskiness for general gambles.

\begin{theorem}
\begin{enumerate}
  \item If $E \log(1+X)<0$, then $\lambda_\infty<1$ and $\lambda_\infty$ is the unique positive solution of \rref{FigUniform}.
  \item      If $E \log(1+X)\ge 0$, then $\lambda_\infty=L(X)=1$.
\end{enumerate}\end{theorem}

\begin{proof}
It is easier to prove the converse of the two statements.
Let us start with assuming $\lambda_\infty<1$. In that case, the sequence
$$Z_n = \log\left(1+\lambda_n X_n\right)$$ is uniformly bounded. Indeed,
$$ -\infty < \log(1-\lambda_\infty) \le Z_n \le \log(1+|X|) \le |X|\in L^1 \,. $$
As we have $Z_n \to \log(1+\lambda_\infty X)$ a.s., we can then invoke Lebesgue's dominated convergence theorem to conclude
$$ 0 = \lim E Z_n = E \lim Z_n = E \log(1+\lambda_\infty X)\,.$$
In particular, the equation \rref{EqnDefRiskiness} has a positive solution $\lambda_\infty<1$. As $\phi$ is strictly concave and strictly positive on $(0,\lambda_\infty)$, we conclude that we must have $\phi(1)=E \log(1+X)<0$. This proves the second claim.

Now let us assume $\lambda_\infty=1$. In that case, we cannot use Lebesgue's theorem. However, the sequence
$$Z_n^\prime = - \log\left(1+\lambda_n X_n\right)$$ is bounded from below by $- \log(1+|X|) \geq -|X| \in L^1$. We can then apply Fatou's lemma to conclude
$$- E \log(1+X)=E \lim  Z_n^\prime \le  \liminf - E \log\left(1+\lambda_n X_n\right)=0\,,$$
or
$$ E \log(1+X) \ge 0\,.$$ This proves the first claim.
\end{proof}

After stating our main theorem, we are eventually able to define our extended Foster--Hart measure of riskiness.
\begin{definition}
\label{ExtRiskiness}
Let $X$ be a gamble with maximal loss $L$.
If $E \log(1+X/L)<0$, we define the extended Foster--Hart measure of riskiness $\rho(X)$ as the unique positive solution of the equation
$$  E \log\left(1+\frac{ X}{\rho(X)}\right) = 0\,.$$
If $E \log(1+X/L) \ge 0$, we define $\rho(X)$ as the maximal loss of $X$,
$$\rho(X)= L\,.$$
\end{definition}


\section{Conditional Riskiness and No Bankruptcy}

\subsection{Learning and Conditional Riskiness}

So far we discussed the riskiness in a static framework as it is done in \cite{FoHa09}. However, in many situations it is more useful to measure risk in a dynamic way. If there is, for instance, more information about an asset available, this should be visible in the assigned risk of this position. A dynamic framework allows to take such updating information into account. The purpose of this section is to embed our extended Foster--Hart riskiness into a dynamic framework.

In the following, let $(\Omega,\cF,(\cF_t)_{t\in \mathbb{N}},P)$ be a filtered probability space, where the filtration $(\cF_t)_{t\in \mathbb{N}}$ represents the information structure given at the respective time $t$. We denote by $\cA_t$ the set of all $\cF_t-$measurable random variables and consider a sequence of random variables $(X_t)$ that is adapted to the filtration $(\cF_t)_{t\in \mathbb{N}}$. In order to be able to measure the risk of $X_t$ in every time period $s<t$, $X_t$ has to satisfy all the conditions which define a gamble given the respective information structure $(\cF_s)$.
\begin{definition}
  We call a random variable $X$ on $(\Omega,\cF,P)$ a gamble for the $\sigma$--field $\cF_s\subset \cF$ if $X$ is bounded and satisfies  $E[X|\cF_s]>0$ $a.s.$ and $P(X_t<0|\cF_s)>0$ $a.s.$.
\end{definition}

In the following, we assume that for $t>s$, $X_t$ is a gamble for  $\cF_s$.

As time goes by, we learn something about the realization of the random variable and are therefore able to quantify the risk more precisely. Measuring the risk of $X_t$ in every single time period $s<t$ yields a family of conditional risk measures $(\rho_s(X_t))_{s=1...t-1}$, where every $\rho_s(X_t)$ is a $\cF_{s}-$measurable random variable that is either the solution of the equation
\begin{equation}\label{EqnDefDynRiskiness}
  E \left[ \log\left(1+\frac{X_{t}}{\rho_s(X_t)}\right) | \cF_{s} \right] = 0  \,
\end{equation}
or $\rho_s(X_t)=L_s(X_t)$, where we define
\[
L_s(X_{t}):=\essinf\{Z\in \cA_s | P(-X_{t}>Z | \cF_{s})=0 \hspace{2mm} a.s.\}.
\]
This means, $L_s(X_{t})$ is the maximal loss of $X_t$ given the information at time $s$.
As we have done for our extended measure of riskiness in the static case, we set $\rho_s(X_t)=L_s(X_t)$ on the set $$B:=\left\{E \left[ \log\left(1+\frac{X_{t}}{L_s(X_t)}\right) |  \cF_s \right] \ge 0 \right\}$$ and define $\rho_s(X_t)$ as the solution of equation \rref{EqnDefDynRiskiness} on $B^c$.
We call the so defined family of conditional risk measures $(\rho_s(X_t))_{s=1...t-1}$ the \textit{dynamic extended Foster--Hart riskiness} of $X_t$.

The next Lemma shows that the dynamic extended Foster--Hart riskiness is indeed well defined.
\begin{lemma} \label{welldef}
There exists one and only one $\cF_s$--measurable random variable $\rho_s(X_t) \ge L_s(X_t)$ that solves equation \rref{EqnDefDynRiskiness} on $B^c$ and satisfies $\rho_s(X_t)=L_s(X_t)$ on $B$.
\end{lemma}

The proof is given in the appendix.

\begin{remark}
An important question that arises in a dynamic framework is how the conditional risks at different times are interrelated. This question leads to the important notion of time--consistency. A dynamic risk measure ($\rho _{s})_{s=1...t-1}$ is called time--consistent
if for any gamble $X_t^1,X_t^2$ and for all $t=1,...,t-2$ it holds that
$$\rho_{s+1}(X_t^1)\geq \rho _{s+1}(X_t^2) \hspace{2mm} a.s. \Longrightarrow \rho _{s}(X_t^1)\geq \rho_{s}(X_t^2) \hspace{2mm} a.s.$$
That means, in particular, that if we know that tomorrow in every state of the world gamble $X_t^1$ is assigned to have a higher risk than gamble $X_t^2$, this should also hold true today.

This desirable property is not satisfied by the dynamic extended Foster--Hart riskiness. A detailed example which illustrates that is given in the appendix.
\end{remark}

\subsection{No Bankruptcy}

We conclude our investigation by checking  that our extended definition of riskiness avoids bankruptcy.

\begin{theorem}
\label{NoBankruptcy}
  Let $(X_n)$ be a sequence of gambles that are uniformly bounded above by some integrable random variable $Y>0$ and satisfy some minimal possible loss requirement, i.e. there exists $\epsilon>0$ such that $a.s.$ $$L_{n-1} (X_n) \ge \epsilon>0$$ for all $n$.
  Let $W_0>0$ be the initial wealth and define recursively
  $$W_{t+1}=W_t + X_{t+1}$$ if
  $E \left[ \log\left(1+\frac{X_{t+1}}{W_t}\right) | \cF_t \right] \geq 0$ and
  $$W_{t+1}=W_t$$ else.
  We then insure no bankruptcy, i.e.
  $$P [\lim W_t=0 ] = 0\,.$$
\end{theorem}

\begin{proof}
  Note first that $W_t >0$ $a.s.$. This can be shown by induction. We have $W_0>0$. We have either  $W_{t+1} = W_{t}$ which is positive by induction hypothesis,  or $W_{t+1} = W_{t}+X_{t+1}$. In this case, the condition $E \left[ \log\left(1+\frac{X_{t+1}}{W_t}\right) | \cF_{t} \right] \geq 0$ implies that a.s.
  \[
   W_t \geq \rho_{t}(X_{t+1}) \geq L_t(X_{t+1}).
  \]

   Thus, $W_t-L_t(X_{t+1})\geq 0$ $a.s.$. If $\rho_{t}(X_{t+1}) = L_t(X_{t+1})$, then we have $P(X_{t+1}=L_t(X_{t+1}) |\cF_t)=0$. Hence, $a.s.$ it holds that
  \[
   W_{t+1}>W_t-L_t(X_{t+1})\geq 0.
  \]
  We can thus define
  $S_t = \log W_t$. We claim that $S$ is a submartingale. Indeed, on the set $$A:=\left\{E \left[ \log\left(1+\frac{X_{t+1}}{W_t}\right) |  \cF_t \right] < 0 \right\}$$
  which belongs to $\cF_t$, there is nothing to show. On the set $A^c$, we have
  \begin{align*}
    E \left[ S_{t+1} | \cF_t \right] &=  E \left[ \log W_{t+1} | \cF_t \right] \\
    &=  \log W_t + E \left[ \log \frac{W_{t+1}}{W_t}  | \cF_t \right] \\
    &= \log W_t + E \left[ \log \left(1+\frac{X_{t+1}}{W_t}\right)  | \cF_t \right] \\
    &\geq \log W_t = S_t \,.
  \end{align*}

  $S$ is thus a submartingale. We apply a theorem on submartingale convergence to complete our proof. We will use Theorem 1 in \cite{Sh84}, Chapter VII.
  For $a>0$, let $\tau_a=\inf\left\{t \ge o : X_t >a \right\}$. A stochastic sequence belongs to class ${\cal C}^+$ if for every $a>0$ we have
  $$E \left(X_{\tau_a}-X_{\tau_a-1}\right)^+ 1_{\left\{\tau_a < \infty \right\}} < \infty.$$

  Let us check that our sequence $S$ is of class ${\cal C}^+$. Indeed, we have
  $$ \left(S_{\tau_a}-S_{\tau_a-1}\right)^+ = \log\left(1+\frac{X_{\tau_a}}{W_{\tau_a-1}}\right) 1_{\left\{X_{\tau_a} \ge 0\right\}} $$
  and in that case $W_{\tau_a-1} \ge \rho_{\tau_a-1}(X_{\tau_a}) \geq \epsilon >0$ $a.s$, so we conclude
  $$ E\left(S_{\tau_a}-S_{\tau_a-1}\right)^+ \le  E\log\left(1+\frac{Y}{\rho_{\tau_a-1}(X_{\tau_a})}\right)\le
  E\log\left(1+\frac{Y}{\epsilon}\right) \le E \frac{Y}{\epsilon} < \infty$$
  where $Y$ is the uniform integrable upper bound for our gambles and $\epsilon$ is the minimal possible loss lower bound.
  By Theorem 1 in \cite{Sh84}, Chapter VII, we conclude that the set $\{S_t \to -\infty\}$ is a null set. (Indeed, on the set $\{S_t \to -\infty\}$, $S$ is bounded above. The theorem then states that the limit of $S$ exists and is finite (almost surely), and thus cannot be negative infinity.)
\end{proof}

\section{A List of Examples}

In order to gain further inside into the behavior of our generalized measure of riskiness, we consider in this section different continuous distributed gambles and compute their riskiness.

\subsection{Uniform Distribution}

Firstly, we consider as in Example \ref{MotivatingEx} an uniform distributed random variable over the interval $[-100,b]$. In order to get a positive expectation, we must have $b>100$. In Figure \ref{Uniform} the graph of the riskiness $\rho$ (Figure \ref{uniform1}) as well as the solution $\lambda$ of equation \rref{EqnDefRiskiness} (Figure \ref{uniform2}), which is the reciprocal of the riskiness, is plotted against the maximal gain $b$ of the gambles. Example \ref{MotivatingEx} already shows that not for all values $b$ there exists a positive solution of the defining equation \rref{EqnDefRiskiness}. Indeed, the critical value for which $E \log\left(1+\lambda^* X \right) = 0\,$ is $b_c\approx171.83$. For values of $b$ which are greater than $b_c$ there is no positive solution of equation \rref{EqnDefRiskiness} and we take in this case $\rho=L(X)=100$ as the riskiness.

\begin{figure}
  \centering
  \subfigure[Generalized Foster--Hart riskiness $\rho$]{
    \label{uniform1}
    \includegraphics[width=0.46\textwidth]{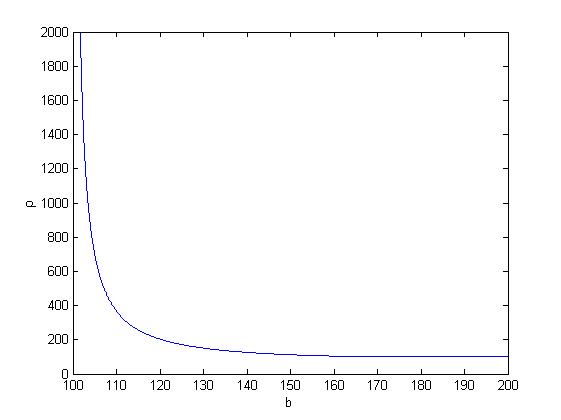}
  }
  \hfill
  \subfigure[$\lambda=\frac{1}{\rho}$]{
  \label{uniform2}
  \includegraphics[width=0.46\textwidth]{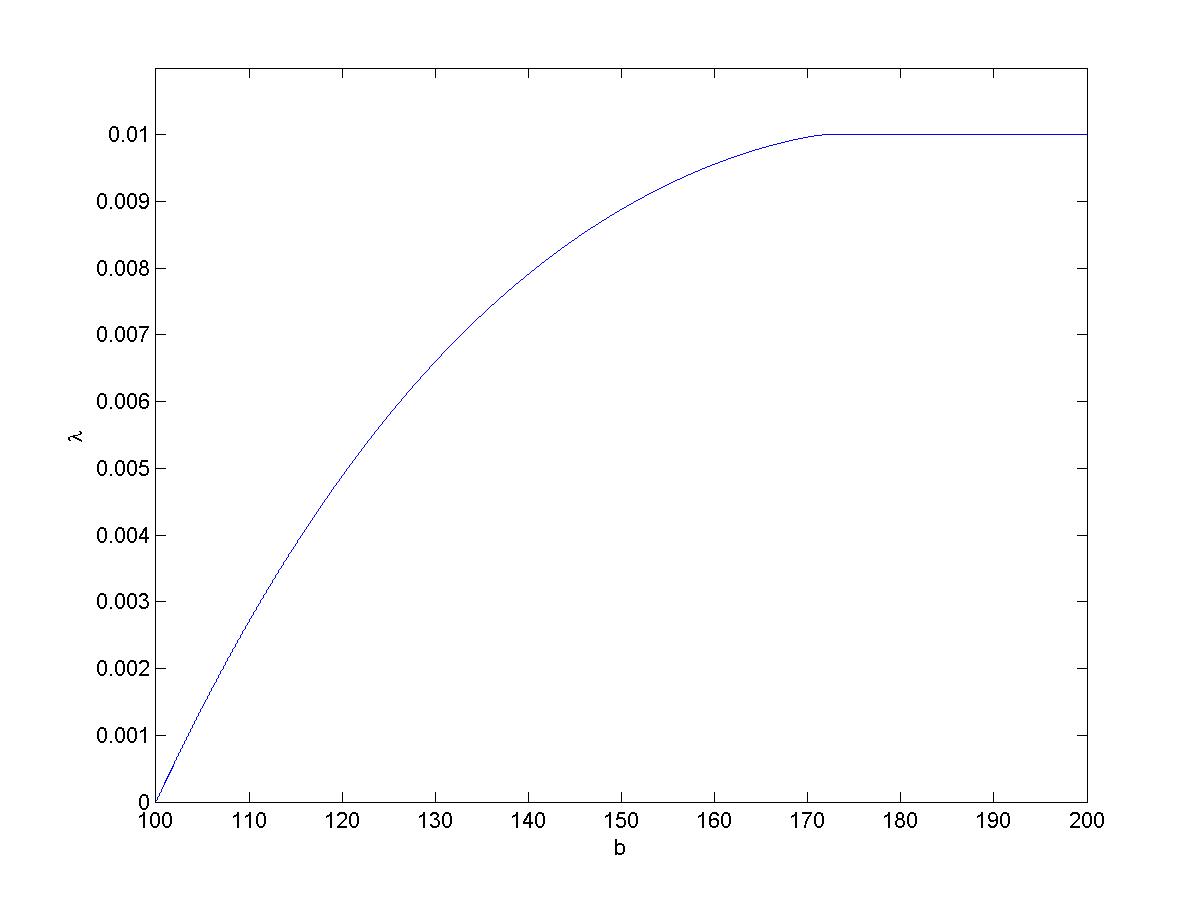}
  }
  \caption{$\rho$ and $\lambda$ for uniform distributed gambles over $[-100,b]$.}
  \label{Uniform}
\end{figure}

As a result the graph of the riskiness yields a continuous function; the riskiness tends to the maximal loss $L=100$ as we approach the critical value $b_c$ and converges to infinity as the expectation of the gamble goes to $0$ (i.e. $b\downarrow 100$).

\subsection{Beta Distribution}

\begin{figure}
  \centering
  \subfigure[Generalized Foster--Hart riskiness $\rho$]{
    \label{beta1}
    \includegraphics[width=0.46\textwidth]{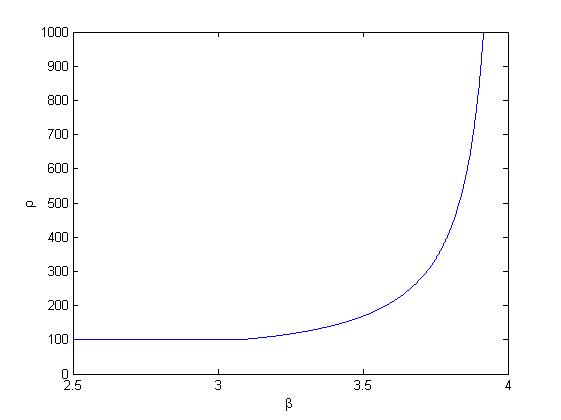}
  }
  \hfill
  \subfigure[$\lambda=\frac{1}{\rho}$]{
  \label{beta2}
  \includegraphics[width=0.46\textwidth]{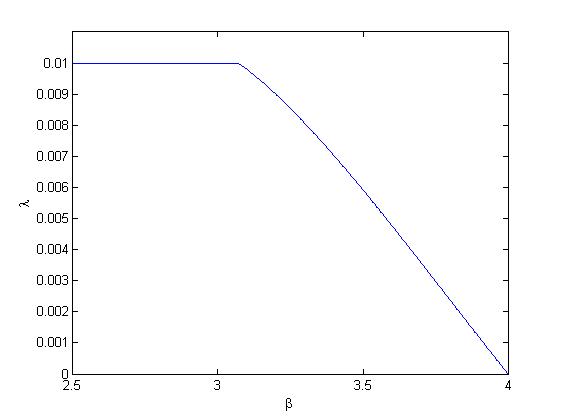}
  }
  \caption{$\rho$ and $\lambda$ for beta distributed gambles over $[-100,200]$ with $\alpha=2$.}
  \label{Beta}
\end{figure}

Next, we consider beta distributed gambles. The density of random variable $X$ that is beta distributed over the compact interval $[-a,b]$ is, for instance, given in \cite{JoKo95} as
\[
\varphi(x;\alpha,\beta,a,b)=\frac{1}{B(\alpha,\beta)}\frac{(x-a)^{\alpha-1}(b-x)^{\beta-1}}{(b-a)^{\alpha+\beta-1}}, x\in[-a,b], \alpha,\beta >0,
\]
where $B(\alpha,\beta)$ denotes the Betafunction defined as
\[
B(\alpha,\beta)=\int^{1}_{0}{t^{\alpha-1}(1-t)^{\beta-1}dt}.
\]
The mean of $X$ is given by
\[
E[X]=\frac{\alpha b+\beta a}{\alpha +\beta}.
\]
In our example we fix the interval as $[-100,200]$ and choose $\alpha=2$. We compute the riskiness as a function of the parameter $\beta$. Again, in order to be able to do the computation, the mean has to be greater than zero which implies $\beta<4$.
We let the value for $\beta$ run up to $4$, i.e. the expectation goes to zero.

Figure \ref{Beta} shows the graph of the riskiness and of $\lambda=\frac{1}{\rho}$ against the values of $\beta$. The defining equation \rref{EqnDefRiskiness} has a positive solution for all $\beta$ greater than $\beta_c\approx3.05$. For $\beta<\beta_c$, where no positive solution exists, the maximal loss $L=100$ is used to determine the riskiness. Again a continuous function is plotted that goes to infinity as the expectation tends to zero.

\subsection{Lognormal Distribution}

\begin{figure}
  \centering
  \subfigure[Generalized Foster--Hart riskiness $\rho$]{
    \label{lognormal2}
    \includegraphics[width=0.46\textwidth]{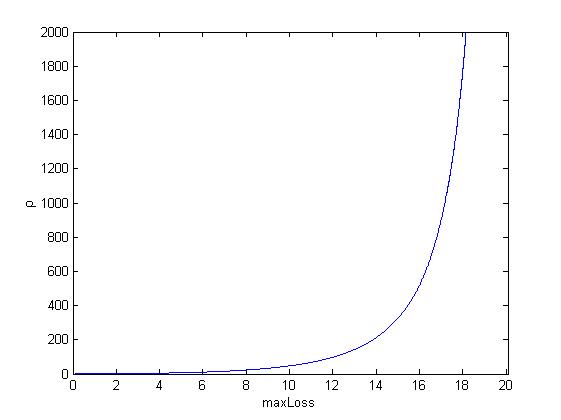}
  }
  \hfill
  \subfigure[$\lambda=\frac{1}{\rho}$]{
  \label{lognormal1}
  \includegraphics[width=0.46\textwidth]{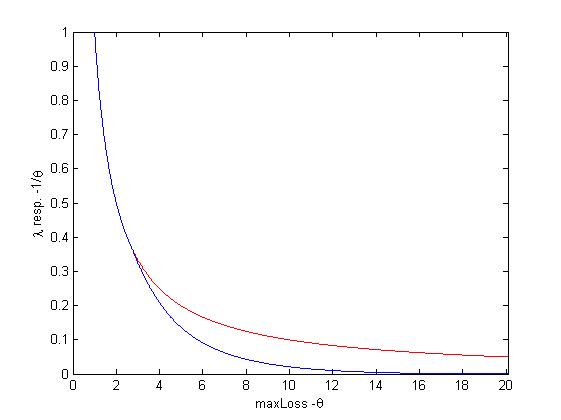}
  }
  \caption{$\rho$ and $\lambda$ for lognormally distributed gambles over $(\theta,\infty)$ with $\mu=1$ and $\sigma=2$.}
  \label{lognormal}
\end{figure}

The lognormal distribution is used in many financial applications, for instance in the widely used Black--Scholes options pricing model. Therefore, it seems to be important to be able to apply the measure of riskiness for this distribution.

A random variable $X$ is said to be lognormally distributed if its density $\varphi$ is
\[
\varphi(x;\mu, \sigma, \theta)=\frac{1}{(x-\theta)\sqrt{2\pi\sigma}}\exp(-\frac{1}{2}\frac{(log(x-\theta)-\mu)^2}{\sigma^2}), x>\theta, \footnote{See \cite{JoKo95}.}
\]
where $\mu$ and $\sigma$ are the expected value, respectively the standard deviation, of the normally distributed random variable $X^N=\log(X-\theta)$ and $-\theta$ is the maximal loss of $X$.

We fix $\mu=1$ and $\sigma=2$ and let $\theta$ run through the interval $[-20.08,0)$. (We have to take $\theta \geq -20.08$ in order to make sure that the expectation is positive). The critical value for which there exists no zero for the defining equation is $\theta\approx -2.72$. For any $\theta\geq -2.72$ we take therefore the maximal loss $-\theta$ as the riskiness.

We plot the graph of the riskiness in Figure \ref{lognormal}. Figure \ref{lognormal1} shows the graph of $\lambda=\frac{1}{\rho}$ as well as of $\frac{1}{-\theta}$, against the maximal loss $L(X)=-\theta$. The upper curve shows $\frac{1}{-\theta}$. If $-\theta< 2.72$, both lines coincides and for $-\theta \geq 2.72$, i.e. where a solution of the defining equation exists, we always have that $\lambda<L(X)$. As the expectation goes to zero, i.e. the maximal loss gets closer to $20.08$, the riskiness tends to infinity or in other words $\lambda$ converges to zero.

As in the examples before, the figure yields a continuous riskiness function. This illustrates that our generalized Foster--Hart measure of riskiness provides a continuous extension of \cite{FoHa09}, which allows to apply the measure for many more classes of random variables, without losing the operational interpretation of avoiding bankruptcy in the long run.

\section{Conclusion}

For many continuous distributed gambles the equation which defines the measure of riskiness proposed by \cite{FoHa09} does not have a finite solution. We extended their defintion by using the maximal loss of the gamble as its riskiness, whenever the defining equation does not have a solution. Our extension allows to apply the measure of riskiness also for continuous distributed gambles which have a strictly positive density on their support and which are bounded from below. We have justified our extension by showing both that it is indeed a consistent extension of the Foster--Hart measure of riskiness and that the no-bankruptcy result of \cite{FoHa09} carries over to the extended measure of riskiness.

Finally, as it is important for many financial applications to measure risk dynamically, i.e. taking updating information into account, we defined the extended riskiness also in a dynamic framework.

\newpage
\begin{appendix}
\section{Appendix}

  The appendix contains some technical details and proofs.

\subsection*{Approximation by Discrete Random Variables}

\begin{proof3}{\ref{LemmaIncreasingApprox}}

 Suppose the support of $X$ is $[-L,M]$. Obviously we have $x_k^n<x_{k+1}^{n}$.

 Further,
 \[
   x^n_k = -1 + \frac{k}{2^n} (M+1)= -1 + \frac{2k}{2^{n+1}} (M+1)=x^{n+1}_{2k}
 \]
  and therefore
 \[
   X_n=\sum_{k=0}^{2^n-1} x^n_k 1_{\left\{x^n_k \le X < x^n_{k+1}\right\}}=\sum_{k=0}^{2^n-1} x^{n+1}_{2k} 1_{\left\{x^{n+1}_{2k} \le X <    x^{n+1}_{2(k+1)}\right\}}.
 \]
  The last expression is smaller than $X_{n+1}$ since we have on each interval $[x^{n+1}_{2k},x^{n+1}_{2(k+1)})\neq \emptyset$ that
 \[
   X_n=x^{n+1}_{2k}
 \]
  and
 \[
   X_{n+1}\geq x^{n+1}_{2k}
 \]
 with $X_{n+1} > x^{n+1}_{2k}$ on $\emptyset \neq [x^{n+1}_{2k+1},x^{n+1}_{2(k+1)})\subset [x^{n+1}_{2k},x^{n+1}_{2(k+1)})$.
 Hence, $(X_n)$ is an increasing sequence.

 Analogously, we can prove the statement for the support $[-1,\infty)$, just by replacing $x_k^n$ by $-1+\frac{k}{2^n}$ and $X_n$ by $$\sum_{k=0}^{n2^n-1} x^n_k 1_{\left\{x^n_k \le X < x^n_{k+1}\right\}}+(-1+n) 1_{\{X\geq(-1+m)\}}.$$

 Finally, we have $\lim X_n = X$ $a.s.$ by construction of the sequence $(X_n)$.
\end{proof3}

\subsection*{Time--Consistency}

 Consider two discrete gambles $X_2^1$ and $X_2^2$ that have their payments in two periods $(t=2)$ from now. They are distributed according to the binomial trees given below. In $t=1$ two states of the world are possible which occur with equal probability $\frac{1}{2}$. We compute the riskiness today $(t=0)$ and in each state in $t=1$.

Gamble $X_2^1$ has the following structure: \\
\tikzstyle{level 1}=[level distance=5cm, sibling distance=2.5cm]
\tikzstyle{level 2}=[level distance=4.5cm, sibling distance=1.2cm]

\tikzstyle{bag} = [text width=4em, text centered,shape=circle,draw=blue!50,fill=blue!20]
\tikzstyle{end} = [circle, minimum width=3pt,fill, inner sep=0pt]

\begin{tikzpicture}[grow=right, sloped]
\node[bag] {$\rho(X_2^1) \approx219.426$}
    child {
        node[bag] {$\rho_1^2(X_2^1) =250$}
            child {
                node[end, label=right:
                    {$-200$}] {}
                edge from parent
                node[above] {}
                node[below]  {$\frac{1}{2}$}
            }
            child {
                node[end, label=right:
                    {$1000$}] {}
                edge from parent
                node[above] {$\frac{1}{2}$}
                node[below]  {}
            }
            edge from parent
            node[above] {}
            node[below]  {$\frac{1}{2}$}
    }
    child {
        node[bag] {$\rho_1^1(X_2^1) =120$}
        child {
                node[end, label=right:
                    {$-100$}] {}
                edge from parent
                node[above] {}
                node[below]  {$\frac{1}{2}$}
            }
            child {
                node[end, label=right:
                    {$600$}] {}
                edge from parent
                node[above] {$\frac{1}{2}$}
                node[below]  {}
            }
        edge from parent
            node[above] {$\frac{1}{2}$}
            node[below]  {}
    };
\end{tikzpicture}

\tikzstyle{level 1}=[level distance=5cm, sibling distance=2.5cm]
\tikzstyle{level 2}=[level distance=4.5cm, sibling distance=1.2cm]

\tikzstyle{bag} = [text width=4em, text centered,shape=circle,draw=blue!50,fill=blue!20]
\tikzstyle{end} = [circle, minimum width=3pt,fill, inner sep=0pt]

Hence, $X_2^1$ has the payoffs $\{600,-100,1000,-200\}$ occurring each with equal probability. The riskiness computed in $t=1$ in state one is $\rho_1^1(X_2^1)=120$ and in state two $\rho_1^2(X_2^1)=250$. The riskiness today is $\rho(X_2^1)\approx 219.426$.

The second gamble $X_2^2$ is distributed according to the following tree: \\

\begin{tikzpicture}[grow=right, sloped]
\node[bag] {$\rho(X_2^2) \approx 243.76$}
    child {
        node[bag] {$\rho_1^2(X_2^2) =250$}
            child {
                node[end, label=right:
                    {$-240$}] {}
                edge from parent
                node[above] {}
                node[below]  {$\frac{1}{2}$}
            }
            child {
                node[end, label=right:
                    {$6000$}] {}
                edge from parent
                node[above] {$\frac{1}{2}$}
                node[below]  {}
            }
            edge from parent
            node[above] {}
            node[below]  {$\frac{1}{2}$}
    }
    child {
        node[bag] {$\rho_1^1(X_2^2) = 120$}
        child {
                node[end, label=right:
                    {$-105$}] {}
                edge from parent
                node[above] {}
                node[below]  {$\frac{1}{2}$}
            }
            child {
                node[end, label=right:
                    {$840$}] {}
                edge from parent
                node[above] {$\frac{1}{2}$}
                node[below]  {}
            }
        edge from parent
            node[above] {$\frac{1}{2}$}
            node[below]  {}
    };
\end{tikzpicture}

Although the payoffs of $X_2^2$ differ from the payoffs of $X_2^1$, the riskiness numbers at time $t=1$ coincides. Today, however, the risk of $X_2^2$ is strictly higher than the risk of $X_2^1$. That contradicts the time--consistency condition. Therefore, the dynamic extended Foster--Hart riskiness is not time--consistent.

\subsection*{Existence of Conditional Riskiness}

\begin{proof3}{\ref{welldef}}
Let $X_t$ be a gamble for the $\sigma$--field $\cF_s$.
Without loss of generality, we can assume $L_s(X_t)=1$ almost surely (else replace $X_t$ by $X_t / L_s(X_t)$).

We write $\cG=\cF_s$ and $X=X_t$ in the following for shorter notation.
  We fix a regular version $P(\omega,dx)$ for the conditional probability distribution of $X$ given $\cG$ (which exists as $X$ takes values in a Polish space). Whenever we write conditional expectations or probabilities in the following, we have this regular version in mind.

  For our construction, we need that there are wealth levels $W$ for which we accept the gamble $X$ given $\cG$.
  We thus start with the following observation.

  \begin{lemma}
    There exist $\cG$--measurable random variables $W\ge 1$ such that
    $$ E [ \log(W+X)| \cG ] \ge \log W \,.$$ In particular, this holds true for all $W$ with
    $$ W \ge \frac{ 2 E[X^2  | \cG ]}{E[X|\cG]}$$ and
    $$ \left| \frac{X}{W} \right| \le \frac{1}{2} \,.$$
  \end{lemma}

  \begin{proof}
    We use the estimate
    \begin{equation}
      \label{IneqLog}
      \log (1+x) \ge x - 2 x^2   \end{equation}
 for $|x| \le 1/2$ (which one can obtain from a Taylor--expansion and the Lagrange version of the error term). Take an $\cG$--measurable $W$ with $$ W \ge \frac{ 2 E[X^2  | \cG ]}{E[X|\cG]}$$ and
    $$ \left| \frac{X}{W} \right| \le \frac{1}{2} \,.$$ Such $W$ exists because $X$ has finite variance. We can simply take $$W=2 \max\left\{ \frac{  E[X^2  | \cG ]}{E[X|\cG]} , 2 |X| \right\}\,.$$
    As $|X/W| \le 1/2$, $\log(1+X/W)$ is everywhere defined. By the estimate \rref{IneqLog}, we obtain
    \begin{align*}
      E\left[ \log \left(1+\frac{X}{W}\right) | \cG \right] & \ge
      E\left[ \frac{X}{W} - \frac{2X^2}{W^2} | \cG \right]\\
      &= \frac{1}{W} \left( E\left[X | \cG \right] -  \frac{ E\left[  2X^2 | \cG \right]}{W}\right)\,,
    \end{align*}
    and now we can use the fact that $ W \ge \frac{ 2 E[X^2  | \cG ]}{E[X|\cG]}$ to conclude that
    \begin{align*}
      E\left[ \log \left(1+\frac{X}{W}\right) | \cG \right] & \ge 0\,.
    \end{align*}

  \end{proof}

As a consequence of the preceding lemma, the set
$$\Lambda = \left\{ \lambda \,\,\cG-\mbox{measurable} | \, 0<\lambda\le 1, E \left[ \log(1+\lambda X\right) | \cG \right] \ge 0 \}$$ is not empty. Let $\lambda_0$ be the $\cG$--essential supremum of $\lambda$. By definition, $\lambda_0$ is $\cG$--measurable and $\lambda_0 \ge \lambda$ for all $\lambda\in \Lambda$. Moreover, $\lambda_0$ is the smallest $\cG$--measurable random variable with these properties.

The set $\Lambda$ is  upwards directed: take $\lambda_1,\lambda_2\in\Lambda$. Then we have for $\lambda_3=\max\{\lambda_1,\lambda_2\}$
\begin{align*}
  E \left[ \log(1+\lambda_3 X) | \cG \right] &=
  1_{\{\lambda_1\ge \lambda_2\}} E \left[ \log(1+\lambda_1 X) | \cG \right] \\
  &+ 1_{\{\lambda_1< \lambda_2\}} E \left[ \log(1+\lambda_2 X) | \cG \right] \ge 0 \,.
\end{align*}
    The other properties being obvious, we conclude $\lambda_3 \in \Lambda$.
    Hence, $\Lambda$ is upwards directed; as a consequence, there exists a sequence $(\lambda_n)\subset \Lambda$ with $\lambda_n \uparrow \lambda_0$.

    Our next claim is $ E \left[ \log(1+\lambda_0 X) | \cG \right] \ge 0$.
 The sequence
$$Z_n = - \log\left(1+\lambda_n X\right)$$ is bounded from below by $- \log(1+|X|) \geq -|X| \in L^1$. We can then apply Fatou's lemma to conclude
$$- E \left[\log(1+\lambda_0 X) | \cG \right] =E \left[\lim  Z_n | \cG \right]  \le   \liminf - E \left[\log\left(1+\lambda_n X_n\right)| \cG \right] \le 0\,,$$
or
$$ E\left[ \log(1+\lambda_0 X) | \cG \right]\ge 0\,.$$

We claim now that we have \begin{equation}
  E \left[ \log\left(1+\lambda_0 X\right) | \cG \right] =0
\end{equation}
on the set $\{\lambda_0< 1\}$. This will conclude the proof of our lemma.

It is enough to establish the claim on all sets
$$ \Gamma_n = \left\{ \lambda_0 \le 1-\frac{1}{n} \right\}$$ for all $n \in \mathbb N$.
From now on, we work on this set only without stating it explicitly.
Let
$$A_{m,n} = \left\{ E \left[ \log\left(1+\lambda_0 X\right) | \cG \right] \ge \frac{1}{m} \right\}
\cap \Gamma_n \,.$$ We will show that $A_{m,n}$ is a null set for all $m,n\in \mathbb N$.

 Let $\epsilon=\frac{1}{1+mn}$ and set $\lambda_1=(1-\epsilon)\lambda_0+\epsilon$. Then we have $\lambda_1>\lambda_0$ and $\lambda_1 \le (1-\epsilon)(1-1/n) + \epsilon =1-1/n+\epsilon/n<1$. We also note \begin{equation}
   \label{EqnHelpPlus}
   \lambda_1-\lambda_0=\epsilon(1-\lambda_0) \le \epsilon.
 \end{equation}

We have
\begin{equation}
  \label{EqnHelpStar}
  1+\lambda_1 X \ge  1-\lambda_1 \ge \frac{1-\epsilon}{n}>0\,.
\end{equation}
$\log(1+\lambda_1X)$ is thus finite (on $\Gamma_n$ where we work).

We now want to show
$$  E \left[ \log\left(1+\lambda_1 X\right) | \cG \right] \ge 0 $$
on $A_{m,n}$. If $A_{m,n}$ was not a null set, this would contradict the definition of $\lambda_0$ as the $\cG$--essential supremum of $\Lambda$.

In order to establish the desired inequality, it is enough to show
$$  E \left[ \log\left(1+\lambda_1 X\right) | \cG \right]
-
 E \left[ \log\left(1+\lambda_0 X\right) | \cG \right]
\ge -\frac{1}{m} $$ because of the definition of $A_{m,n}$.
Now, on the set $\{X\ge 0\}$ we have $\log\left(1+\lambda_1 X\right) \ge \log\left(1+\lambda_0 X\right)$.

We need a uniform estimate for $\log\left(1+\lambda_1 X\right)-\log\left(1+\lambda_0 X\right)$ on the set $\{X<0\}.$
With the help of the mean value theorem, we obtain on $\{X<0\}$
$$  \log\left(\frac{1+\lambda_1 X}{1+\lambda_0 X}\right)  \ge - \frac{n}{1-\epsilon} \left(\lambda_1-\lambda_0 \right) \ge  - \frac{n\epsilon}{1-\epsilon}\,.$$
(How can one see this: by the mean value theorem and \rref{EqnHelpPlus}, we have
$$\log\left(1+\lambda_1 X\right)-\log\left(1+\lambda_0 X\right)
= \frac{1}{\xi} (\lambda_1-\lambda_0) X$$ for some $\xi $ in between $1+\lambda_1 X$ and $1+\lambda_0X$.
 By \rref{EqnHelpStar} $1+\lambda_1 X \ge (1-\epsilon)/n$. Hence, we have $0<1/\xi \le \frac{n}{1-\epsilon}$.)

 By the definition of $\epsilon$, we thus have
 $$  \log\left(\frac{1+\lambda_1 X}{1+\lambda_0 X}\right)  \ge - \frac{n\epsilon}{1-\epsilon} = -\frac{1}{m}$$ uniformly on $\{X<0\}$
 as desired. It follows
 \begin{align*} &  E \left[ \log\left(1+\lambda_1 X\right) | \cG \right]
-
 E \left[ \log\left(1+\lambda_0 X\right) | \cG \right]
 \\
&\ge E \left[ \log\left(\frac{1+\lambda_1 X}{1+\lambda_0 X}\right)1_{\{X<0\}} | \cG \right] \\
 &\ge
 -\frac{1}{m} \,.
 \end{align*}
\end{proof3}

\end{appendix}

\addcontentsline{toc}{section}{References}
\bibliography{paper1}


\end{document}